%% file: PaperRAMICS.tex
\documentclass[runningheads]{llncs}

\usepackage[utf8]{inputenc}

\usepackage{mathtools}
\usepackage{amsfonts}
\usepackage{amssymb}

\usepackage{xcolor}

\newcommand\logeq{\mathrel{\vcentcolon\Leftrightarrow}}

\newcommand{\cat}{%
	\mathbf %
}

\newcommand{\domain}[ 1 ]{%
	\mathrm{dom}(#1)%
}

\newcommand{\codomain}[ 1 ]{%
	\mathrm{cod}(#1)%
}

\newcommand{\idarrow}[ 1 ]{%
	\mathbf{1}_{#1}%
}

\newcommand{\stern}{\textbf{$\star$}}

\newcommand{\struc}[1]{\mathcal{#1}}

\begin{document}

\title{Computer-supported Exploration of a Categorical Axiomatization of Modeloids}
\titlerunning{Categorical Axiomatization of Modeloids}

\author{Lucca Tiemens\inst{1} \and Dana S. Scott\inst{2} \and Christoph Benzmüller\inst{3,4} \and \\ Miroslav Benda\inst{5}}	
\authorrunning{L. Tiemens, D. S. Scott, C. Benzmüller and M. Benda}	
	
\institute{Technische Universität Berlin, Berlin, Germany\\ 
	\email{tiemens@campus.tu-berlin.de} \and
University of California, Berkeley, USA \\
\email {scott@andrew.cmu.edu} \and
%	Carnegie Mellon University, Pittsbergh, USA  \and	 
	Freie Universität Berlin, Berlin, Germany \and 
	University of Luxembourg, Esch-sur-Alzette, Luxembourg  \\
\email{c.benzmueller@fu-berlin.de}
\and
	Formerly at University of Washington, Seattle, USA \\
\email{miro@oasa.com}
}
	
\maketitle
\begin{abstract}
	A modeloid, a certain set of partial bijections, emerges from the idea to abstract from a structure to the set of its partial automorphisms. It comes with an operation, called the derivative, which is inspired by Ehrenfeucht-Fraïssé games. In this paper we develop a generalization of a modeloid first to an inverse semigroup and then to an inverse category using an axiomatic approach to category theory. We then show that this formulation enables a purely algebraic view on Ehrenfeucht-Fraïssé games.
\end{abstract}

\section{Introduction}

Modeloids have been introduced by M. Benda \cite{Benda1979}. They can be seen as an abstraction from a structure to a partial automorphism semigroup created in the attempt to study properties of structures from a different, more general angle which is independent of the language that is defining the structure. We do not follow Benda's original formulation in terms of an equivalence relation but treat modeloids as a certain set of partial bijections. Our recent interest in them was triggered by D. Scott's suggestion to look at the modeloidal concept from a categorical perspective. The new approach aims at establishing a framework in which the relationship between different structures of the same vocabulary can be studied by means of their partial isomorphisms. The overall project is work in progress, but as a first result we obtained a purely algebraic formulation of Ehrenfeucht-Fraïssé games.

 %Recent interest in them arose when D. Scott suggested to look at the concept of a modeloid form a categorical perspective. The idea is motivated as follows. Modeloids are seen as an abstraction from an algebraic structure created in the attempt to study properties of structures from a different, more general angle. Our new approach aims at establishing a framework in which the relationship between different algebraic structures of same vocabulary can be studied by means of their partial isomorphisms. It is a work in progress but one of the results we already obtained is a purely algebraic formulation of Ehrenfeucht-Fraïssé games.

Throughout the project, computer-based theorem proving is employed in order to demonstrate and explore the virtues of automated and interactive theorem proving in context.  The software used is Isabelle/HOL~\cite{nipkow2002isabelle} in the 2019 Edition. We are generally interested in conducting as many proofs of lemmas and theorems as possible by using only the \textit{sledgehammer}\footnote{\textit{Sledgehammer}~\cite{Sledgehammer} is linking interactive proof development in Isabelle/HOL with anonymous calls to various integrated automated theorem proving systems. Among others, the tool converts the higher-order problems given to it into first-order representations for the integrated provers, it calls them and analyses their responses, and it tries to identify minimal sets of dependencies for the theorems it proves this way.} command,
 and to study how far full proof automation scales in this area. Reporting on these practically motivated studies, however, will not be the focus of this paper. We only briefly mention here how we encoded, in Isabelle/HOL, an inverse semigroup and an inverse category, and we present a summary of our practical experience.

Inverse semigroups (see e.g. \cite{lawson1998} for more information) play a major role in this paper. They serve as a bridge between modeloids and category theory. The justification for this is given by the fact that an inverse semigroup can be faithfully embedded into a set of partial bijections by the Wagner-Preston representation theorem. This opens up the possibility of generalizing modeloids, which are sets of partial bijections, to the language of inverse semigroup theory.

Once there, we have a natural transition from an inverse semigroup to an inverse category (for further reference see \cite{Linckelmann2012}). We introduce the theory of inverse categories by an equational axiomatization that enables computer-supported reasoning. This serves as the basis for our formulation of a categorical modeloid.

In each stage of generalization the derivative, a central operation in the theory of modeloids, can be adapted and reformulated. This operation is about extending the elements of a modeloid. Suppose that $\tau$ is a finite relational vocabulary meaning that $\tau$ consists only of finitely many relation and/or constant symbols.
As it turns out, the derivative on a categorical modeloid on the category of finite $\tau$-structures is equivalent to playing an Ehrenfeucht-Fraïssé game. 

This paper is organized in the following way. In section \ref{sc: Modeloid} we define both modeloids and the derivative operation. We then turn to inverse semigroups in section \ref{sc: InverseSemigroupAndModeloid} and develop the axiomatization of a modeloid in inverse semigroup language. Section \ref{sc: CategoricalAxiomatizationModeloid} shows how to represent a category in Isabelle/HOL and defines the categorical modeloid. After the derivative operation is established in this context, we give an introduction to Ehrenfeucht-Fraïssé games in Section \ref{sc: AlgebraicEFGames} and present the close connection between the categorical derivative and Ehrenfeucht-Fraïssé games. Proofs for the stated theorems, propositions and lemmas are presented in the extended preprint \cite{ARXIV} of this paper (cf.~also \cite{BachelorTiemens}); the Isabelle/HOL source files are available online.\footnote{See \url{http://christoph-benzmueller.de/papers/RAMICSadditionalMaterial.zip}.}

\section{Modeloids}\label{sc: Modeloid}

Let us first recall the definitions of a partial bijection and of partial composition.

\begin{definition}[Partial bijection and partial composition]
	A partial bijection $f:X \to Y$ is a partial injective function. The inverse of $f$, also a partial bijection and denoted by $f^{-1}$, is given by the preimage of the elements in the codomain of $f$: $f^{-1}(y) = f^{-1}(\{y\}), \, \forall y \in \text{cod} (f).$
	  
	The composition between two partial functions $f:X \to Y$ and $g:Y \to Z$ is defined only on $f^{-1}(dom (g) \cap cod (f) )$. Then the partial composition
	\[(g\circ f) (x) = g(f(x)), \quad \forall x \in f^{-1}(dom (g) \cap cod (f) )\]
	is well-defined.
\end{definition}

Furthermore, let $\Sigma$ be a finite non-empty set. We then define 
\begin{equation} 
	F(\Sigma) := \{f:\Sigma \to \Sigma \mid f \textnormal{ is a partial bijection} \}  
\end{equation}
as the set of all partial bijections on $\Sigma$.
% It is the set $F(\Sigma)$ that we define a modeloid on.

\begin{definition}[Modeloid \cite{Benda1979}] \label{def: Modeloid}
	Let $M \subseteq F(\Sigma)$. $M$ is called a modeloid on $\Sigma$ if, and only if, it satisfies the following axioms:
	\begin{enumerate}
		\item Closure of composition: $f,g \in M \Rightarrow f \circ g \in M$
		\item Closure of taking inverses: $f \in M \Rightarrow f^{-1} \in M$
		\item Inclusion property: $f \in M$ and $A \subset dom (f)$ implies $f|_A \in M$ 
		\item Identity: $id_\Sigma \in M$
	\end{enumerate}
\end{definition} 

As such, a modeloid is a set of partial bijections which is closed under composition and taking inverses, which has the identity on $\Sigma$ as a member,  and which satisfies the inclusion property. The inclusion property can be seen as a downward closure in regards of function restriction.

In order to further illustrate the definition, we present a motivating example from model theory.

\begin{example}
	Let $S = (A, R_1,...)$ be a finite relational structure. The set $M$ of all partial isomorphisms on $S$ forms a \emph{modeloid}.
\end{example}
The name modeloid originates from the above example since $ S $ is also called a model. For further motivation, background information and details on modeloids, we refer to Benda's paper~\cite{Benda1979}; a nice example in there is the construction of a Scott Sentence presented through modeloidal glasses \cite[p. 82]{Benda1979}. We, on the other hand, turn to the core concept of the derivative which is defined in the following way. For convenience we represent a partial bijection as a set of tuples.

\begin{definition}[Derivative]\label{def: derivative}
	Let $M$ be a modeloid on $\Sigma$. Then the derivative $D(M) \subseteq F(\Sigma)$ is defined by
	\begin{multline*}
	\{(x_1,y_1),..,(x_{n},y_{n})\} \in D(M) \logeq \\
	\forall a \in \Sigma \, \exists b \in \Sigma: \{(x_1,y_1),..,(x_{n},y_{n}),(a,b)\} \in M \, \wedge \\
	\forall a \in \Sigma \, \exists b \in \Sigma: \{(x_1,y_1),..,(x_{n},y_{n}),(b,a)\} \in M
	\end{multline*}
\end{definition}

A derivative $D(M)$ is thus a set which only contains partial bijections that can be extended by an arbitrary element from $\Sigma$ and which then still belong to $M$. This extension can take place either in the domain or in the range of the function. The next two results \cite[Prop 2.3]{Benda1979} provide some insight into why modeloids and the derivative operation are in harmony.

\begin{lemma}
	Let $M$ be a modeloid on $\Sigma$ and $D(M)$ the derivative. Then we have that $D(M) \subseteq M$.
\end{lemma}

\begin{proposition}
	If $M$ is a modeloid then so is $D(M)$.
\end{proposition}

The importance of these results is essentially due to the fact that they enable us to apply the derivative repeatedly.

\section{Inverse Semigroups and Modeloids}\label{sc: InverseSemigroupAndModeloid}

In this section we show how the Wagner-Preston representation theorem justifies our generalization of a modeloid to inverse semigroup language. We also discuss how well proof automation performs in the context of inverse semigroups. Some familiarity with the Isabelle/HOL proof assistant~\cite{nipkow2002isabelle,Sledgehammer} is assumed.

\subsection{Inverse Semigroups in Isabelle/HOL}

We start with the equational definition of an inverse semigroup.

\begin{definition}[Inverse semigroup \cite{Howie1995}]\label{def: invSemi}
	Let $S$ be a set equipped with the binary operation $*:S \times S \to S$ and 
	the unary operation $a \mapsto a^{-1}$. $(S,^{-1},*)$ is called an inverse semigroup 
	if, and only if, it satisfies the axioms
	\begin{enumerate}
		\item $(x * y) * z = x * (y * z)$ for all $x,y,z \in S$,  
		\item $x * x^{-1} * x = x$ for all $x \in S$, 
		\item $(x^{-1})^{-1} = x$ for all $x \in S$ and
		\item $x * x^{-1} * y * y^{-1} = y * y^{-1} * x * x^{-1}$ for all $x,y \in S$
	\end{enumerate}
\end{definition}

%show representation in Isabelle/HOL - summaries how result was proven with this
%and outlook what else worked

An inverse in semigroup theory is a generalization of the known group theoretical definition. This generalized definition does not depend on a specified unique neutral element. Intuitively, it can be thought of as the inverse map of a partial bijection.

\begin{definition}[Inverse]
	Let $(S, *)$ be a semigroup and $x \in S$. Then $y \in S$ is called an inverse of $x$ if, and only if, $ x * y * x = x \text{ and } y * x * y = y$. 
\end{definition}

We encode an inverse semigroup as follows in Isabelle/HOL. %For more information see the Isabelle Manual

%here goes an image of the implementation in Isabelle
{\flushleft
\includegraphics[width=\textwidth]{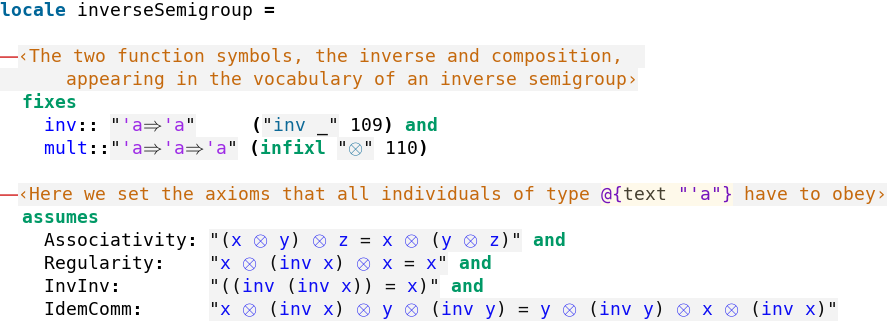}
}

The domain for individuals is chosen to be $'a$, which is a type variable. This means we have encoded a polymorphic version of inverse semigroups.

Using this implementation almost all results needed for proving the Wagner-Preston representation theorem, which we will discuss shortly, can be found by automated theorem proving. Occasionally, however, some additional lemmas to the ones usually presented in a textbook (e.g.~\cite{lawson1998})  are needed. By automated theorem proving we here mean the use of \emph{sledgehammer}~\cite{Sledgehammer} for finding the proofs of the given statements without any further interaction. Regarding equivalent definitions of an inverse semigroup, we were able to automate the proofs of the following theorem (except for $2. \Rightarrow 1.$, which is due to a Skolemization issue).

%Theorem about equivalent definitions (result)

\begin{theorem}[\cite{lawson1998}] \label{thm: defSemi}
	Let $(S,*)$ be a semigroup. Then the following are equivalent:
	\begin{enumerate}
		\item There is $^{-1}: S \to S$ such that $(S,^{-1},*)$ is an inverse semigroup.
		\item Every element of $S$ has a unique inverse.
		\item Every element of $ S $ is regular, meaning $\forall x \in S \, \exists y \in S: x*y*x = x$, and idempotents in $S$ commute. 
	\end{enumerate}
\end{theorem}

Our experiments confirm that automated theorem proving (and also model finding) can well support the exploration of an axiomatic theory as presented. However, the intellectual effort needed to model and formulate the presented mathematics in the first place is of course still crucial, and a great deal of work has gone into this intuitive aspect of the development process. A more technical challenge also is to find suitable intermediate steps that can be proven automatically by \textit{sledgehammer}.
%The additional lemmas which are talked about in the above paragraph have to be discovered by intuition which builds on the concept of try and error. 

\subsection{Modeloid as Inverse Semigroup}

We now show that every modeloid $M \subseteq F(\Sigma)$ under partial composition is an inverse semigroup. We make use of Theorem \ref{thm: defSemi} by using the third characterization. For this task regard $(M,\circ)$ as a semigroup. This is clear since composition of partial functions is associative. Since the partial identities of $M$ are exactly the idempotent elements in $(M,\circ)$, commutativity is ensured by referring to the next proposition. Furthermore, also by using the next proposition, the closure of taking inverses required by a modeloid implies regularity for all elements in $ M $. Hence, $(M,^{-1},\circ)$ is an inverse semigroup.

% properties of partial bijections

\begin{proposition}[\cite{lawson1998}] \label{prop: PropertiesPartBij}
	Let $f: X \to Y$ be a partial bijection.
	\begin{enumerate}
		\item For a partial bijection $g: Y \to X$, the equations $f = fgf$ and $g = gfg$ hold if, and only if, $g = f^{-1}$
		\item $1_A1_B = 1_{A \cap B} = 1_B1_A$ for all partial identities $1_A$ and $1_B$ where $A,B \subseteq X$
	\end{enumerate}
\end{proposition}

Not only is every modeloid an inverse semigroup, but by the Wagner-Preston representation theorem also every inverse semigroup can be faithfully embedded into $F(\Sigma)$, which is itself a modeloid. This motivates the idea of  formulating the axioms for a modeloid in inverse semigroup language. Our aim is to restate the derivative operation in this context. In order to achieve this, we shall translate the axioms from Definition \ref{def: Modeloid}, examining them one by one.

\begin{enumerate}
	\item Closure of Composition: Because of the embedding, the composition of 
	partial functions will simply be the $*$-operation in an inverse semigroup. 
	\item Closure of taking inverses:  By Theorem \ref{thm: defSemi} an inverse semigroup is such
	that the inverse exists for every element and is unique, hence resembling the inverses
	of partial functions and in particular the closure property.
	\item The inclusion property: Here it is not immediately apparent how this can be expressed
	within an inverse semigroup. We shall see that the \emph{natural partial order} is capable of doing that.
	\item The identity on $\Sigma$: The identity is a certain idempotent element in an inverse
	semigroup. It will lead us to the notion of an \emph{inverse monoid}.
\end{enumerate}

It is Axiom 3 that we focus our attention on next. We define the \emph{natural partial order} and present the Wagner-Preston representation theorem, which establishes a connection to function restriction in $F(\Sigma)$. We introduce notation for such a restriction. For two partial functions $f,g$ we write $g \subseteq f$ to say that $dom(g) \subseteq dom(f)$ and $\forall x \in dom(g): g(x) = f(x)$.

% natural partial order
% wagner preston

\begin{definition}[Natural partial order]
	Let $\Sigma = (\Sigma,^{-1},*)$ be an inverse semigroup and $s,t \in \Sigma$. Then we define for some idempotent $e \in \Sigma$ \[ s \, \textbf{$\leq$} \, t \logeq s = t*e. \]
\end{definition}

\begin{theorem}[Wagner-Preston representation theorem \cite{lawson1998}] \label{thm: wagner-preston}
	Let $\Sigma = (\Sigma,^{-1},*)$ be an inverse semigroup. Then there is an injective homomorphism $\Omega: \Sigma \to F(\Sigma)$, such that for $a,b \in \Sigma$ we have
	$ a \, \textbf{$\leq$} \, b \Longleftrightarrow \Omega(a) \subseteq \Omega(b)$. 
\end{theorem}

From this theorem it is clear what we mean by a faithful embedding of an inverse semigroup into the set of partial bijections $F(\Sigma)$.
Faithfulness corresponds to the fact that the \emph{natural partial order} in light of the representation theorem is equivalent to the partial order which function restriction defines.
This nicely opens up the possibility to capture the essence of the inclusion property from Definition \ref{def: Modeloid} by the \emph{natural partial order}. Let $M \subset S$, where $S$ is an inverse semigroup. Then the inclusion property can be stated as
\begin{equation} \label{eq: incPropSemi}
\forall f \in M \, \forall g \in S:  g \leq f \Longrightarrow g \in M.
\end{equation}

By setting $S = F(\Sigma)$, the dependency of $M$ on $F(\Sigma)$ can be seen explicitly. In the abstract formulation of a modeloid we will keep this subset property.
It is immediate that a modeloid, seen as an inverse semigroup, fulfills (\ref{eq: incPropSemi}) by the following proposition.

\begin{proposition} \label{prop: <subsetEq}
	Let $M$ be a modeloid on $\Sigma$. Then, for $f,g \in M$, we have  \[ g \leq f \Longleftrightarrow g \subseteq f \]
\end{proposition}

In a modeloid $M$ the inclusion property implies that the empty partial bijection, which we denote by $0$, is also included in $M$. As a result we want to establish a similar behavior in the generalized modeloid. The deeper reason for this is found in the definition of the derivative operation, because it requires the notion of an atom, which can only be defined if a zero element is present. Seeing $M$ as an inverse semigroup, $0$ is an idempotent element for which the following property holds: $\forall x \in M: 0 * x = 0$. Hence, we will call the idempotent with this property the \emph{zero element}. When defining a modeloid in semigroup language we require the zero element to be part of it.

Turning to Axiom 4, which is $id_\Sigma \in M$, we examine which element of an inverse semigroup $S$ is most suitable for this task. To evaluate, we again look at the modeloid $M$ as an inverse semigroup. In this semigroup, $id_\Sigma$ will be an idempotent $e$ satisfying $\forall x \in M: e * x = x$. Such an element is known as a neutral element in the context of group theory. We require for the inverse semigroup, which we eventually call a modeloid, that $e$ is part of it. What we get is known as an inverse monoid in the literature.

\begin{remark}
	Given an inverse monoid, denoted by $S^1$, and the element $e$ with $e * x = x, \, \forall x \in S^1$. Consider the representation theorem again: this theorem is not guaranteeing uniqueness of the embedding, and in fact there can be several ones. Hence, we cannot assume that $e$ will be mapped to the identity $id_\Sigma$. However, for all idempotent $f \in S^1$ we have that $f \leq e$ because $ f = e * f $ by the assumption about $e$. 
	Hence, $e$ is always the upper bound of all idempotents in $S^1$.  
\end{remark}

We have prepared everything needed for defining a modeloid again. We shall call it a semimodeloid. Note, as mentioned before, that a modeloid is a subset of $F(\Sigma)$ for some non-empty set $\Sigma$ and, as discussed, we keep this subset property to state the inclusion axiom.

\begin{definition}[Semimodeloid]
	Let $S^1 = (\Sigma, \, ^{-1},*,e,0)$ be an inverse monoid. Then $M \subseteq \Sigma$ is called a semimodeloid if, and only if,
	\begin{enumerate}
		\item $\forall x,y \in M: (x * y) \in M$
		\item $\forall x \in M: x^{-1} \in M$
		\item $\forall x \in M \, \forall y \in S^1: y \leq x \Rightarrow y \in M$
		\item $e \in M$
	\end{enumerate}
\end{definition}

\begin{remark}
	A semimodeloid is again an inverse monoid with the zero element. 
\end{remark}

\begin{proposition}
	Every semimodeloid can be faithfully embedded into a modeloid.
	Furthermore, by the considerations above, every modeloid is a semimodeloid.
\end{proposition}

% derivative

Now we develop the derivative operation in the setting of a semimodeloid. Consider again Definition \ref{def: derivative} in which we have introduced the derivative operation. It is evident that the elements of $\Sigma$ are of crucial importance. Furthermore, we are required to be able to extend the domain of a function by one element at a time. This poses a challenge because in an inverse monoid this information is not directly accessible. But as we shall see, it is possible to obtain.

First we characterize the elements of $\Sigma$. Therefore, consider $F(\Sigma)$ and realize that all the singleton-identities $id_{\{a\}}$, for $a \in \Sigma$,  are in natural bijection to the elements of $\Sigma$. The special property of such a singleton-identity is  
\begin{equation} 
	\forall f \in F(\Sigma): f \subseteq id_{\{a\}} \Rightarrow (f = id_{\{a\}} \vee f = 0) 
\end{equation}
 since $dom(id_{\{a\}}) = \{a\}$. Seeing $F(\Sigma)$ as an inverse monoid with zero element leads to the following definition.
\begin{definition}[Atom]
	Let $S^1$ be an inverse monoid with zero element $0$. Then a non-zero element $x \in S^1$ is an atom if, and only if, \[ \forall f \in S^1: f \leq x \Rightarrow (f = x \vee f = 0) \]
\end{definition}
Our plan is to use the notion of an atom to define the derivative. The next lemma justifies this usage.

\begin{lemma}
	The idempotent atoms in $F(\Sigma)$ are exactly the singleton-identities.
\end{lemma}

This suffices to define the derivative for semimodeloids. We then ensure that the definition matches Definition \ref{def: derivative} if the semimodeloid is on $F(\Sigma)$.

\begin{definition}[Derivative---semimodeloid] \label{def: derivativeSemiMode}
	Let $M$ be a semimodeloid on the inverse monoid $S^1$ with zero element $0$. We define the derivative $D(M)$ of $M$ as 
	\begin{multline*}
	D(M) := 
	\{ f \in M \, | \, \forall \text{ idempotent atoms } a \in S^1 \, \exists x \in M : (f \leq x \, \wedge \, a \leq x^{-1}x) \, \wedge \\
	\forall \text{ idempotent atoms } b \in S^1 \, \exists y \in M : (f \leq y \, \wedge \, b \leq yy^{-1})\}
	\end{multline*}
\end{definition}
If we think about $x$ in the above definition as a partial bijection, then $x^{-1}x$ is the identity on the domain of $x$ and, hence, the condition $a \leq x^{-1}x$ expresses that $a$ is in the domain of $x$. Similarly $b \leq yy^{-1}$ states that $b$ is in the range of $y$.

%This point of view is justified by Theorem \ref{thm: wagner-preston}.

\begin{proposition}
	The derivative on a modeloid $M$ produces the same result as the semimodeloidal derivative on $M$.
\end{proposition} 

%With this result we conclude this section and move on to the categorical setting.

\section{Categorical Axiomatization of a Modeloid}\label{sc: CategoricalAxiomatizationModeloid}

We use an axiomatic approach to category theory based on \emph{free logic} \cite{Scott1967,lambert1960definition,lambert2002free}. As demonstrated by Benzmüller and Scott \cite{J40}, this approach enables the encoding of category theory in Isabelle/HOL. Their encoding work is extended below to cover also inverse categories. Subsequently we formalize modeloids and derivatives in this setting.

\subsection{Category Theory in Isabelle/HOL}\label{subsc: CategoryTheoryInIsabelle}

When looking at the definition of a category $\cat C$, one can realize that the objects $A,B,C,..$ are in natural bijection with the identity morphisms $\idarrow A, \idarrow B, \idarrow C,...$ because those are unique. This enables a characterization of a category just by its morphisms and their compositions, which is used to establish a formal axiomatization. However, in this axiomatic approach we are faced with the challenge of partiality, because the composition between two morphisms $f,g \in \cat C$ is defined if, and only if,
\begin{equation} \label{eq: compPart}
\domain g = \codomain f.
\end{equation}
As a result composition is a partial operation.

An elegant way to deal with this issue is by changing the underlying logic to \emph{free logic}. In free logic an explicit notion of existence is introduced for the objects in the domain that we quantify over. In our case the domain consists of the morphisms of a category. The idea now is to define the composition total, that is,  any two morphisms can always be composed, but only those compositions ``exist'' that satisfy (\ref{eq: compPart}). Because we can distinguish between existing and non-existing morphisms, we are able to formulate statements that take only existing morphisms into account. In this paper we want to work with a unique non-existing morphism which will be denoted by $\star$. Hence a composition of morphisms, that does not satisfy (\ref{eq: compPart}), will result in $\star$.
% Due to the achievement of finding a shallow embedding of free logic in Isabelle/HOL by Benzmüller and Scott, first order axiomatic category theory could also be implemented. 
We refer to Benzmüller and Scott \cite{J40} for more information on the encoding of free logic in Isabelle/HOL.

Based upon this groundwork, a category in Isabelle/HOL is defined as follows.

%picture of category in Isabelle/HOL having detailed explanation

{\flushleft
\includegraphics[width=\textwidth]{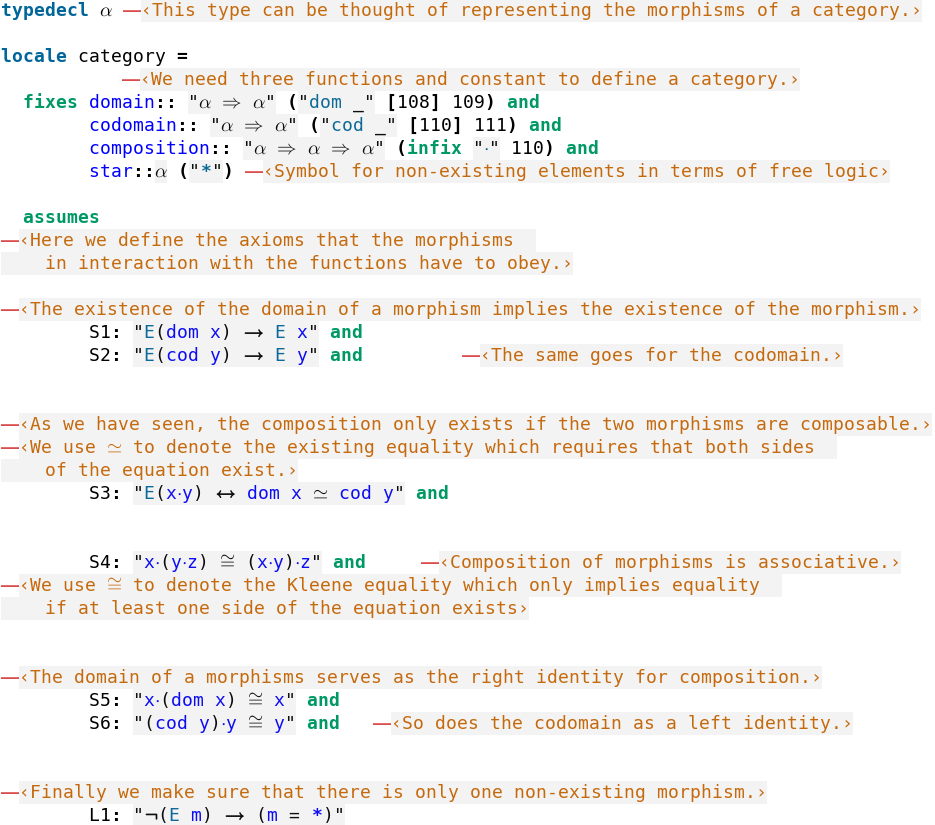}
}

For convenience, we will assume a category to be small for the rest of this paper. As a result, a category for us has only a set of morphisms which satisfies the above axiom schema. This allows us to use notation from set theory. We write $(m: X \to Y) \in \cat C$ to mean that $m$ is a morphism from the category $\cat C$. In addition, it says that $\domain{m} \cong X$ and $\codomain{m} \cong Y$, so $X$ is the domain of $m$ and $Y$ the codomain. The identity morphisms $X$ and $Y$, which are representing objects in the usual sense, are characterized by the property that $X \cong \domain{X} \cong \codomain{X}$, respectively for $Y$. Hence every $c \in \cat C$ satisfying $c \cong \domain{c}$ or $c \cong \codomain{c}$ is representing an object, and we refer to such a morphism as an \emph{object}. 

We want a categorical generalization of an inverse semigroup, so let's turn to the question of how to introduce generalized inverses to a category.
In the above setting we found that by adding the axioms of an inverse semigroup, which are responsible for shaping these inverses (Definition \ref{def: invSemi}, Axioms 2-4), we arrive at a notion that is equivalent to the usual definition of an inverse category. Note that this definition is adopted to our free logic foundation by using \emph{Kleene equality}, which is denoted by $\cong$. We emphasize again that this equality between terms states that, if either term exists, so does the other one and they are equal. 

\begin{definition}[Inverse category \cite{Kastl1979}]
	A small category $\cat C$ is called an inverse category if for any morphism $s: X \to Y \in \cat C$
	there exists a unique morphisms $\hat{s}: Y \to X$ such that $s \cong s \cdot \hat{s} \cdot s$ and 
	$\hat{s} \cong \hat{s} \cdot s \cdot \hat{s}$.
\end{definition}

%Present the two ways of implementing an inverse category
For the representation in Isabelle/HOL we skolemized the definition.

{\flushleft
\includegraphics[width=\textwidth]{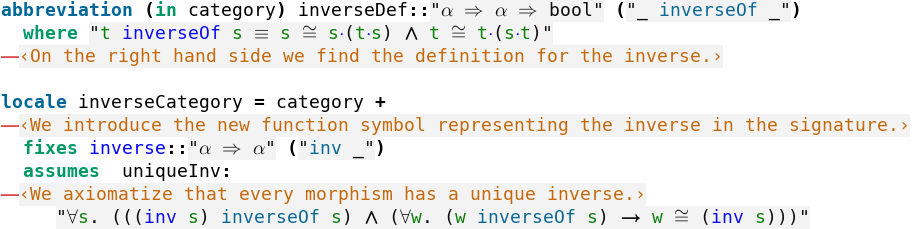}
}

Next,  we see the quantifier free definition.

{\flushleft
\includegraphics[width=\textwidth]{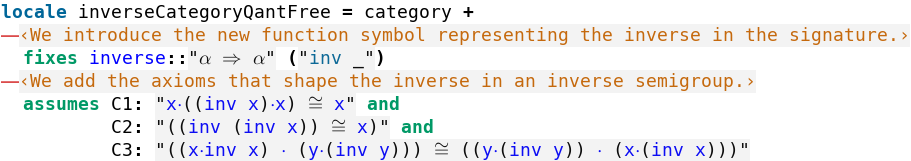}
}

The equivalence between the two formulations has been shown by interactive theorem proving. Again, a significant number of the required subproofs could be automated by \emph{sledgehammer}. In addition, the minimality of the axioms for the quantifier free version above was checked effectively using \emph{sledgehammer} and \emph{nitpick}\footnote{\emph{nitpick}\cite{nitpick} is a counterexample generator for higher-order logic integrated with Isabelle/HOL.}.
Inverse categories are interesting to us because of the following proposition.

\begin{proposition} \label{prop: oneObCatInvSemi}
	Let $\cat C$ be an inverse category with exactly one object. Then $\cat C$ is an inverse semigroup.
\end{proposition}

This allows us to generalize a semimodeloid to an inverse category by formulating the new axioms in such a way that this categorical construction will collapse to a semimodeloid under the condition of having just one object.

\subsection{Categorical Axiomatization of a Modeloid}\label{subsc: CategoricalModeloid}

% Introduce the natural partial order again and idempotence

The notion of the \emph{natural partial order} is also definable in an inverse category. To state it, we first introduce a definition for \emph{idempotence}.

%We redefine the \emph{natural partial order}. We also give a definition for \emph{idempotence} in order to get a feel for the new underlying logic.

\begin{definition}[Idempotence]
	Let $\cat C$ be a small category. Then a morphism $e \in \cat C$ is called idempotent if, and only if,
	\[ e \cdot e \cong e. \]
\end{definition}

Whenever we do not assume that both sides of the equation exist, we use \emph{Kleene equality}.

\begin{definition}[Natural partial order \cite{Linckelmann2012}]
	Let $\cat C$ be an inverse category and let $s,t: X \to Y$ be morphisms in $\cat C$. We define
	\[ s \leq t \logeq \exists \text{ idempotent } e \in End_C(X): s \cong t \cdot e \]
\end{definition}
where $End_C(X) := \{m \in \cat C \, | \, m: X \to X  \}$ is called an endoset.

When defining a categorical modeloid $M$ on an inverse category $\cat C$, we will see that for each object $X$ in $\cat C$, $End_C(X)$ is a semimodeloid.
 We require the category to have a zero element in each of its endosets in order to define an atom. For this we simply write that $\cat C$ has all zero elements.

\begin{definition}[Categorical modeloid] \label{def: catMod}
	Let $\cat C$ be an inverse category with all zero elements. Then a \emph{categorical modeloid} $M$ on $\cat C$ is such that $M \subseteq \cat C$ satisfies the following axioms:
	\begin{enumerate}
		\item $a,b \in M \Rightarrow a \cdot b \in M$
		\item $a \in M \Rightarrow a^{-1} \in M$
		\item $\forall \, a \in \cat C \, \forall b \in M: a \leq b \Rightarrow a \in M$
		\item $\forall \, \text {objects }X \in \cat C: X \in M$
	\end{enumerate}
\end{definition}

It is evident that this definition is close by its appearance to a semimodeloid. However, we are now dealing with a network of semimodeloids and have thus reached a much more expressive definition.

\begin{proposition} \label{prop: endoIsSemimod}
	Let $\cat C$ be an inverse category with all zero elements and let $M$ be a categorical modeloid on $\cat C$. Then for each object $X$ in $M$ we get that $End_M(X)$ is a semimodeloid (on itself).
\end{proposition}

\begin{remark}
	Every semimodeloid can easily be seen as a categorical modeloid by the fact that an inverse monoid with zero element is an one-object inverse category.
\end{remark}

% definitions for derivative

We have formulated a generalization of a modeloid in category theory. What is left now is to define the derivative in this context. We will need the notion of a homset and of an atom, which we already introduced for semigroups.

\begin{definition}[Homset]
	Let $\cat C$ be a small category. Then the homset between two elements $X,Y \in \cat C$, satisfying $X \cong dom(X)$ and $Y \cong dom(Y)$, is defined as
	\[ Hom_C(X,Y) := \{ m \in \cat C \, | \, m: X \to Y \} \]
\end{definition}
Hence an endoset is a special case of a homset. We only assume zero elements to be present in endosets and as a result an atom needs to be part of an endoset.

\begin{definition}[Atom]
	Let $\cat C$ be an inverse category with all zero elements. Then an element $a \in End_C(X)$, for some object $X \in \cat C$, is an \emph{atom} if, and only if, the existence of $a$ implies that $a$ is not the zero element and \[ \forall e \in End_C(X): e \leq a \text{ implies that } e \cong a \vee e \cong 0_{End_C(X)}.  \]
\end{definition}

% derivative

This concludes the preliminaries for defining the derivative on a homset.

\begin{definition}[Derivative---homset]
	Let $\cat C$ be an inverse category with all zero elements and let $M$ be a categorical modeloid on $\cat C$. We define the derivative on $Hom_M(X,Y)$ for $X,Y \in M$ as $D(Hom_M(X,Y)) :=
	\{ f \in Hom_M(X,Y) \, | $\\ 
	$\forall \text{ idempotent atoms } a \in End_M(X) \, \exists h \in Hom_M(X,Y): (f \leq h \wedge a \leq h^{-1}h) \, \wedge$ \\
	$\forall \text{ idempotent atoms } b \in End_M(Y) \, \exists g \in Hom_M(X,Y): (f \leq g \wedge b \leq gg^{-1}) \}$
	% \begin{multline*}
	% D(Hom_M(X,Y)) := \\
	% \{ f \in Hom_M(X,Y) \, | \\ 
	% \forall \text{ idempotent atoms } a \in End_M(X) \, \exists h \in Hom_M(X,Y): (f \leq h \wedge a \leq h^{-1}h) \, \wedge \\
	% \forall \text{ idempotent atoms } b \in End_M(Y) \, \exists g \in Hom_M(X,Y): (f \leq g \wedge b \leq gg^{-1}) \}
	% \end{multline*}
\end{definition}

\begin{remark}
	Let $\cat C$ be an inverse category with just one object $X$ and a zero element. Then $\cat C$ is an inverse semigroup by Proposition \ref{prop: oneObCatInvSemi} and
	the derivative on the homset $D(Hom_C(X,X))$ is equal to the semimodeloidal derivative  $D(\cat C)$.
\end{remark}

Now the key property of this operation is that it produces a categorical modeloid again if we apply it to \emph{all} homsets simultaneously.

% derivative is categorical modeloid again

\begin{theorem}
	Let $\cat C$ be an inverse category with all zero elements and let $M$ be a categorical modeloid on $\cat C$. Then
	\[ \bigcup_{X,Y \in M} D(Hom_M(X,Y)) \] is a categorical modeloid on $\cat C$.
\end{theorem}

As a result we define this to be the derivative operation on categorical modeloids.

\begin{definition}[Derivative---categorical modeloid]
	Let $\cat C$ be an inverse category with all zero elements and let $M$ be a categorical modeloid on $\cat C$. Then we set the derivative as
	\[ D(M) := \bigcup_{X,Y \in M} D(Hom_M(X,Y)). \]
\end{definition}

% At this point we also explore what it means to take the derivative $m$-times because it will be needed in the next section. This is, however, straight forward. 
Let $M$ be a categorical modeloid. We define
\begin{equation}
	D^0(M) := M \text{ and } D^{n+1}(M) := D(D^n(M)) 
\end{equation}
for $n \in \mathbb{N}$. $D^m(M)$ thus takes the derivative $m$-times. This notion is used in the next section. 
%We shall investigate what our established framework is capable of in finite model theory.

%Definition: taking the derivative m-times

\section{Algebraic Ehrenfeucht-Fraïssé games}\label{sc: AlgebraicEFGames}

When moving from classical model theory to the finite case, some machinery for proving inexpressibility results in first-order logic, such as the \emph{compactness theorem}, fails. However, Ehrenfeucht-Fraïssé (EF) games are still applicable and, therefore, play a central notion in finite model theory due to the possibility to show that a property is first-order axiomatizable. For more information see \cite{Libkin2013}. 

In this section we explicitly show that derivatives on categorical modeloids generalize EF games.

\subsection{Rules of EF game}
% recall working of EF game

To play an EF game, two finite $\tau$-structures $\mathcal{A}$ and $\mathcal{B}$, where $\tau$ is a finite relational vocabulary, are needed. In general EF games are not restricted to finite structures, but for our purpose we shall only deal with this case. In order to give an intuitive understanding we imagine two players, which we call the spoiler and the duplicator, playing the game. The rules are quite simple. In $n\in \mathbb{N}$ rounds the spoiler tries to show that the two given structures are not equal, while the duplicator tries to disprove the spoiler every time. A round consists of the following:

\begin{itemize}
	\item The spoiler picks either $\mathcal{A}$ or $\mathcal{B}$ and then makes a move by choosing an element from that structure, so $a\in \mathcal{A}$ or $b\in \mathcal{B}$.
	\item After the spoiler is done, the duplicator picks an element of the other structure and the round ends. 
\end{itemize}

Next we define what the winning condition for each round will be. For convenience let $Part(\struc A,\struc B)$ be the set of all partial isomorphisms from $\mathcal{A}$ to $\mathcal{B}$. Furthermore, given a constant symbole $c$ from $\tau$, we denote by $c^{\mathcal{A}}$ the interpretation of $c$ in the structure $\mathcal{A}$.

\begin{definition}[Winning position \cite{Libkin2013}]
	Suppose the EF game was played for $n$ rounds. Then there are moves $(a_1,..,a_n)$ picked from $\mathcal{A}$ and moves $(b_1,..,b_n)$ picked from $\mathcal{B}$. For this to be a winning position we require that for some $r \in \mathbb{N}$ the map
	\[ \{(a_1,b_1),..,(a_n,b_n),(c^{\mathcal{A}}_1, c^{\mathcal{B}}_1),.., (c^{\mathcal{A}}_r, c^{\mathcal{B}}_r) \} \in Part(\struc A, \struc B)\] where the $c_i$ are all constant symbols of $\tau$.
\end{definition}

In order to win, the duplicator needs to defeat the spoiler in every possible course of the game. We say the duplicator has an \emph{n-round winning strategy in the Ehrenfeucht-Fraïssé game on $\mathcal{A}$ and $\mathcal{B}$} \cite{Libkin2013}, if the duplicator is in a winning position after n moves regardless of what the spoiler does. This is made precise by the back-and-forth method due to Fraïssé.

\begin{definition}[Back-and-forth relation \cite{Ebbinghaus1995}] \label{def: backforth}
	
	We define a binary relation $\equiv_m, \, m\in \mathbb{N}$, on all $\tau$-structures by $\mathcal{A} \equiv_m \mathcal{B}$ iff there is a sequence $(I_j)$ for $0 \leq j \leq m$ such that
	
	\begin{itemize}
		\item Every $I_j$ is a non-empty set of partial isomorphisms from $\mathcal{A}$ to $\mathcal{B}$
		\item (Forth property) $\forall j < m$ we have $\forall a \in \mathcal{A}\, \forall f \in I_{j+1} \, \exists g\in I_j: f \subseteq g \, \wedge \, a \in dom(g)$
		\item (Back property) $\forall j < m$ we have $\forall b \in \mathcal{B}\, \forall f \in I_{j+1} \, \exists g\in I_j: f \subseteq g \, \wedge \, b \in cod(g)$
	\end{itemize}
\end{definition}
Hence $\mathcal{A} \equiv_n \mathcal{B}$ means that the duplicator has a $n$-round winning strategy.

\subsection{The derivative and Fraïssé's method}

% Define category of finite structure

We relate the categorical derivative to Fraïssé's method which we have just seen. In order to do this, we define a categorical modeloid on the category of finite $\tau$-structures, where $\tau$ is a finite relational vocabulary. For that let $\mathcal{A}$ and $\mathcal{B}$ be two finite $\tau$-structures. Denote by $F(\mathcal{A}, \mathcal{B})$ the set \[ \bigcup_{(X,Y)\in \{\mathcal{A}, \mathcal{B}\}^2} Part(X,Y) \] and let $\stern \not \in F(\mathcal{A}, \mathcal{B})$ be an arbitrary element. Then define $C := F(\mathcal{A}, \mathcal{B}) \cup \{\stern\}$. 

We construct two functions $dom: C\to C$ and $cod:C\to C$ such that for a partial isomorphism $f:X\to Y \in F(\mathcal{A}, \mathcal{B})$ we set $dom(f) = id_X$ and $cod(f) = id_Y$, and for the element $\stern$ we define $dom(\stern) = \stern$ and $cod(\stern) = \stern$. 

Next we define a binary operation $\cdot:C\to C$ by
\[ f \cdot g = 
\begin{cases*}
f \circ g,  & if $dom(f) = cod(g)$ and $f,g \neq \stern$\\
\stern,  & else\\
\end{cases*} \] where $\circ$ denotes the composition of partial functions.

\begin{proposition}
	$\cat D := (C,dom,cod,\cdot,\stern,^{-1})$ is an inverse category where $f^{-1}$ denotes the inverse of each partial isomorphism $f$ and $\stern^{-1} = \stern$. The existing elements are exactly all elements in $F(\mathcal{A}, \mathcal{B})$ and the compositions $f \circ g$ in case $dom(f) = cod(g)$, for $f,g \in F(\mathcal{A}, \mathcal{B})$.
\end{proposition}

What we have just seen provides a general procedure for creating categories in our setting, which is founded on a free logic that is itself encoded in Isabelle/HOL.

\begin{corollary}
	$\cat D := (C,dom,cod,\cdot,\stern,^{-1})$ is also a categorical modeloid on itself.
\end{corollary}

\begin{remark}
	Hence we have that every inverse category having a zero element for each of its endosets is also a categorical modeloid and thus admits a derivative. 
\end{remark}

% Theorem about equivalence

At this point we are able to use the derivative on $\cat D$. The final theorem draws the concluding connection between modeloids and Fraïssé's method. We show that in the established setting, an $m$-round winning strategy between $\struc A$ and $\struc B$ is given by the sets which the derivative produces if applied $m$ times. Note the abuse of notation in the way we are using $\equiv_m$ here.

\begin{theorem}
	Let $M$ be the categorical modeloid $\cat D$. Then
	\[ \exists h:X \to Y \in D^m(M) \text{ with } h \neq \stern \Longleftrightarrow X \equiv_m Y, \quad m \in \mathbb{N} \]
\end{theorem}

As a result the derivative on this modeloid is equivalent to playing an EF game between the two structures. Hence on an arbitrary categorical modeloid the derivative can be seen as a generalization of EF games.

\section{Conclusion}

%summary of paper

In this paper we have shown how to arrive at the notion of a categorical modeloid using axiomatic category theory. We started out with a set of partial bijections abstracting from a structure, then we interpreted this set as an inverse semigroup by the embedding due to the Wagner-Preston representation theorem, and, finally, we were able to axiomatize a modeloid in an inverse category. The key feature we employed is the natural partial order which also enabled us to present the derivative operation in each step of abstraction. The categorical derivative on the category of finite structures of a finite vocabulary can then be used to play Ehrenfeucht-Fraïssé games between two structures. As a result a more abstract representation of these games is possible. 

%future work

Using our encoding of inverse categories in Isabelle/HOL, we are currently extending this encoding work to cover also categorical modeloids and their derivatives. This naturally extends the framework established by Benzmüller and Scott so far \cite{J40}.
Furthermore, an investigation of the generalized Ehrenfeucht-Fraïssé games in terms of applicability has to be conducted. We believe that the notion of a categorical modeloid will continue to play a role when connecting model theoretical and categorical concepts.

\subsubsection*{\small Acknowledgment.} {\small We wish to thank the anonymous referees for their helpful comments and suggestions.}

\bibliographystyle{splncs04}

\input{PaperRAMiCS.bbl}
\newpage
\input{Proofs.tex}

\end{document}

%% file: Proofs.tex
\pagebreak
\section*{Appendix}

\setcounter{theorem}{0}
\setcounter{lemma}{0}
\setcounter{proposition}{0}
\setcounter{corollary}{0}

This appendix provides the proofs for all claims we made in the paper. We will go through them in the order in which they appear.

\subsection{Proofs of section \ref{sc: Modeloid}}

\begin{lemma}
	Let $M$ be a modeloid on $\Sigma$ and $D(M)$ the derivative. Then we have that $D(M) \subseteq M$.
\end{lemma}
\begin{proof}
	We want to show $ x \in D(M) \Rightarrow x \in M$. Fix $ x \in D(M) $. Next fix $ a \in \Sigma $. We know we can find $ b \in \Sigma $ such that $ x \cup \{ (a,b) \} \in M$. Since $ M $ is a modeloid, the inclusion property implies $ x \in M $.
\end{proof}

\begin{proposition}
	If $M$ is a modeloid then so is $D(M)$.
\end{proposition}
\begin{proof}
	We will prove the modeloidal axioms for $ D(M) $ in the order in which they appear in the definition of a modeloid.
	
	Hence, we start with the closure of composition. Let $ f,g \in D(M) $. We need to show that for a fixed $ a \in \Sigma $ we can find $ b \in \Sigma $ such that $ (g \circ f) \cup \{(a,b)\} \in M $. The second conjunct of the derivative follows by analogy.
	Therefore, fix $ a \in \Sigma $. We can find $ b_1 \in \Sigma $ such that $ f \cup \{(a,b_1)\} \in M $. Then we can find $ b_2 \in \Sigma $ such that $ g \cup \{(b_1, b_2)\} \in M $. Because $ M $ is closed under the composition, it follows that $ (g \cup \{(b_1, b_2)\}) \circ (f \cup \{(a,b_1)\}) = (g \circ f) \cup \{(a,b_2)\} \in M $. As a result, $ D(M) $ is closed under composition.
	
	Up next is the closure of taking inverses. Let $ f \in D(M) $ and fix $ a \in \Sigma $. Because of the second conjunct of the derivative, we get that for some $ b \in \Sigma $ the statement $ x \cup \{(b,a)\} \in M $ holds. Since $ M $ is a modeloid, this implies that $ x^{-1} \cup \{(a, b)\} \in M $. By analogy we also get $ \forall a \in \Sigma \exists b \in \Sigma: x^{-1} \cup \{(b, a)\} \in M $. Thus $ x^{-1} \in D(M) $.
	
	The inclusion property is evident. Fix $ f \in D(M) $ and a $ g $ satisfying $ g \subseteq f$. \[ \forall a \in \Sigma \exists b \in \Sigma: f \cup \{(a,b)\} \in M\] immediately implies \[ \forall a \in \Sigma \exists b \in \Sigma: g \cup \{(a,b)\} \in M.\] Analogously we get $ \forall a \in \Sigma \exists b \in \Sigma: g \cup \{(b,a)\} \in M$ and hence $ g \in D(M) $.
	
	The fact that $ id_\Sigma \in D(M)$ can be seen by noting $ \forall a \in \Sigma: id_\Sigma \cup \{(a,a)\} = id_\Sigma$ and $ id_\Sigma \in M $. This concludes the proof. 
\end{proof}

\subsection{Proofs of section \ref{sc: InverseSemigroupAndModeloid}}

\setcounter{proposition}{2}

\begin{proposition} \label{prop: <subsetEq}
	Let $M$ be a modeloid on $\Sigma$. Then for $f,g \in M$ \[ g \leq f \Longleftrightarrow g \subseteq f \]
\end{proposition}
\begin{proof}
	Let $\Sigma$ be an alphabet and $M$ a functional modeloid on $F(\Sigma)$. We have already established that $(M, \, ^{-1}, \circ)$ is an inverse semigroup. Fix $f,g \in M$. Supposing that $g \leq f$ holds we know $g = f \circ e$ where $e$ is an idempotent. As such $e$ is a partial identity in $M$. As a result \[ dom(g) = e^{-1}( dom(f) \cap cod(e) ) = dom(f) \cap cod(e) \] and hence $dom(g) \subseteq dom(f)$. Furthermore, $g(x) = (f \circ e)(x) = f(x)$ for $x \in dom(g)$. This yields $g \subseteq f$. \\
	Conversely suppose that $g \subseteq f$. Since $id_{dom(g)} \subseteq id_\Sigma$ we know that $id_{dom(g)} \in M$. In addition, the partial identity is idempotent. Now $f \circ id_{dom(g)} \in M$ and $g = f \circ id_{dom(g)}$ because $dom(g) = dom(f) \cap cod(id_{dom(g)})$ since $dom(g) \subseteq dom(f)$ and $g(x) = f(x) = f(id_{dom(g)}(x))$ for $x \in dom(g)$. As a result $g \leq f$ holds.
\end{proof}

\begin{proposition}
	Every semimodeloid can be faithfully embedded into a modeloid.
	Furthermore, by the considerations above, every modeloid is a semimodeloid.
\end{proposition}
\begin{proof}
	This proposition basically summarizes the work done. By taking the modeloid $ F(\Sigma) $, it is clear that every semimodeloid on $ \Sigma $ can be faithfully embedded into it by the Wagner-Preston representation theorem because the inclusion property only depends on the faithfulness of the embedding with respect to the \emph{natural partial order}.
	Proposition \ref{prop: <subsetEq} establishes that every modeloid is a semimodeloid.
\end{proof}

\begin{lemma}
	The idempotent atoms in $F(\Sigma)$ are exactly the singleton-identities.
\end{lemma}
\begin{proof}
	Idempotent elements in $F(\Sigma)$ are the partial identities. So it suffices to show that $|dom(f)| = 1$ if, and only if, the idempotent $f \in F(\Sigma)$ is an atom. 
	Assume $f \in F(\Sigma)$ is an atom and idempotent. Suppose now that $|dom(f)| > 1$. Then we can find $a,b \in dom(f)$ with $a \not = b$. But then we have that $id_{\{a\}} \subsetneq f$ and $id_{\{a\}} \not = 0$. This is a contradiction to $f$ being an atom. The case $|dom(f) = 0|$ implies that $f = 0$ but an atom is unequal to the zero element. As such that case is also taken care of. \\
	Conversely assume $|dom(f)| = 1$ for some non-zero partial identity $f \in F(\Sigma)$. Then $g \subseteq f$ implies that $dom (g) = \emptyset \vee dom(g) = dom(f)$. If it is the first option we have $g = 0$. And if it is the second we get $g = f$.
\end{proof}

\begin{proposition}
	The derivative on a modeloid $M$ produces the same result as the semimodeloidal derivative on $M$.
\end{proposition} 
\begin{proof}
	Let $M$ be a functional modeloid on $F(\Sigma)$ for an alphabet $\Sigma$. We want to show that the two definitions of the derivative are equivalent in this case. We show that for fixed $f \in M$
	\begin{equation*}
	\begin{gathered}
	\forall a \in \Sigma \, \exists b \in \Sigma: f \cup \{(a,b)\} \in M \\
	\Longleftrightarrow \\
	\forall \text{ idempotent atoms } b \in F(\Sigma) \, \exists x \in M : (f \leq x \, \wedge \, b \leq x^{-1}x)
	\end{gathered}
	\end{equation*}
	because the second part of the definition of the derivative follows by analogy.
	
	To start remember the natural bijection \[ \Psi: \Sigma \simeq \{p \in F(\Sigma) \, | \, p \text{ an atom and idempotent } \}.\]
	Suppose the upper formula holds. Fix $a \in \Sigma$. Then setting $x := f \cup \{(a,b)\}$ yields that $f \leq x$ by proposition \ref{prop: <subsetEq}. Furthermore, $a \in dom(x) = dom(x^{-1}x)$. But then $\Psi(a) \leq x^{-1}x$. Hence we get that $\exists x \in M : (f \leq x \, \wedge \, \Psi(a) \leq x^{-1}x)$. But now we can quantify over $\{p \in F(\Sigma) \, | \, p \text{ an atom and idempotent } \}$ instead of $\Sigma$ which yields the desired result. \\
	Conversely suppose the bottom formula holds. Then fix an idempotent atom $b \in F(\Sigma)$ and let by assumption $x$ be the element with $f \leq x \, \wedge \, b \leq x^{-1}x$. It holds that $x|_{dom(f) \cup \Psi^{-1}(\{b\})} \in M$ because of the inclusion property. But this already yields that \[f \cup \{(\Psi^{-1}(b), x(\Psi^{-1}(b)))\} \in M. \] Quantifying over $\Sigma$ instead of $\{p \in F(\Sigma) \, | \, p \text{ an atom and idempotent } \}$ concludes the proof.   	
\end{proof}

\subsection{Proofs of section \ref{sc: CategoricalAxiomatizationModeloid}}

\setcounter{proposition}{6}	
	
\begin{proposition} \label{prop: endoIsSemimod}
	Let $\cat C$ be an inverse category with all zero elements and $M$ be a categorical modeloid on $\cat C$. Then for each object $X$ in $M$ we get that $End_M(X)$ is a semimodeloid (on itself).
\end{proposition}
\begin{proof}
	Fix an object $X$ in $M$. Once we have proven that $End_M(X)$ is an inverse monoid it will follow that $End_M(X)$ is a semimodeloid. That is because $End_M(X)$ then is closed under composition and taking inverses. Furthermore, the definition of the partial order defined on morphisms will simply reduce to the natural partial order. As a result the inclusion axiom also holds. And at last we have that $X \in End_M(X)$ with the property that $y \cdot X \cong y$ for all $y \in End_M(X)$ hence giving us the neutral element required by a semimodeloid.
	
	Let's now prove that $End_M(X)$ is an inverse monoid. First we will show the closure of the composition. For that fix two elements $a,b \in End_M(X)$. We know that $dom(a) \cong dom(b) \cong cod(a) \cong cod(b) \cong X$. 
	We distinguish two cases. First, we assume that $ X $ exists, so $ X \neq \stern $.
	Therefore, $ dom (a) \simeq cod (b) $ holds and as a result $ a \cdot b $ exists. Hence, we have $ a \cdot b \cdot dom (a \cdot b) \simeq a \cdot b $. This implies 
	\begin{equation}\label{eq: domcod}
		dom (a \cdot b) \simeq cod (dom (a \cdot b))
	\end{equation}
	and the existence of $ b \cdot dom (a\cdot b) $ from which we get $ dom (b) \simeq cod (dom (a \cdot b))$. By using (\ref{eq: domcod}) we deduce $ dom(b) \simeq dom (a\cdot b) $ but this holds $ X \simeq dom(a\cdot b) $. Similarly one obtains $ X \simeq cod (a\cdot b) $. As a result, $ a \cdot b \in End_M(X) $.
	
	Now assume $ X $ does not exist. This yields that $ a\cdot b $ does not exist since otherwise, $ dom(a) $ and $ cod(b) $ would exist which contradicts $ cod(b) \cong dom(a) \cong X $. As a result, also $ dom (a\cdot b) $ and $ cod (a\cdot b) $ do not exist and therefore, $ X \cong dom (a\cdot b) \cong cod (a\cdot b) $ holds since none of the terms is existing. Hence $ a\cdot b \in End_M(X) $.
	
	The closure of inverses is immediate because taken an element $s \in End_M(X)$, $s^{-1} \in M$ by assumption but $dom(s^{-1}) \cong cod(s^{-1}) \cong X$ and as a result $s^{-1} \in End_M(X)$. Now one can regard the inverse function on $End_M(X)$ as a restriction of $(^{-1})$ on $\cat C$. As a result the inverses are unique. Associativity follows by the fact that the composition in $M$ really is the composition in $\cat C$ restricted to $M$. Above we have already taken care of the neutral element. 
\end{proof}
	
\setcounter{theorem}{2}	
\begin{theorem}
	Let $\cat C$ be an inverse category with all zero elements and let $M$ be a categorical modeloid on $\cat C$. Then
	\[ \bigcup_{X,Y \in M} D(Hom_M(X,Y)) \] is a categorical modeloid on $\cat C$.
\end{theorem}	
\begin{proof}
	Assume the assumptions formulated above. Then we define \[ H := \bigcup_{X,Y \in M} D(Hom_M(X,Y)). \] We will prove (1) that all objects from $\cat C$ are in $H$, (2) the closure of taking inverses from H, (3) the closure of the composition on $H$ and (4) the inclusion property hold which is \[\forall \, s \in \cat C \, \forall t \in H: s \leq t \Rightarrow s \in H.\]
	
	\begin{enumerate}
		\item Fix an object $X \in \cat C$. Since $M$ is a categorical modeloid, $X \in M$ and, furthermore $X \in End_M(X)$. We need to prove that $X \in D(End_M(X))$. So fix an idempotent atom $a \in End_M(X)$. We show that $X \leq X$ and $a \leq X^{-1}X$.  
		Note that, since $X \cong dom(X)$, $X$ is idempotent by axiom $S5$ of a category. As a result $X = X\cdot X$ and hence $X \leq X$.
		On the other hand by definition of $End_M(X)$ we have that $a \cong a \cdot X \cong a \cdot X^{-1}X$. Because $a$ is idempotent it commutes with $X^{-1}X$ and hence by definition $a \leq X^{-1}X$. The second part of the condition posed by the derivative holds simply because $X^{-1}X = XX^{-1}$.
		
		\item Next take an element $s \in H$. Then $s \in D(Hom_M(X,Y))$ for some $X,Y \in M$. As a result we can write $s$ as $s:X\to Y$. We know that $s^{-1}:Y\to X \in Hom_M(Y,X)$ and want to show that $s^{-1}:Y\to X \in D(Hom_M(Y,X))$. 
		
		Fix an idempotent atom $a \in End_M(Y)$. Then, since $s \in H$, we find $h \in Hom_M(X,Y): s \leq h \wedge a \leq hh^{-1}$. Also $s^{-1} \leq h^{-1}$ and $h^{-1} \in Hom_M(Y,X)$. Because $ a \leq hh^{-1}$ we know that $a \leq (h^{-1})^{-1}h^{-1}$. In total that yields 
		\begin{align*} 
		\forall \text{ idempotent atoms }& a \in End_M(Y) \\
		&\exists h' \in Hom_M(Y,X): (s^{-1} \leq h' \wedge a \leq h'^{-1}h') 
		\end{align*}
		The second condition required by the derivative follows by an analogous construction. Hence $s^{-1} \in H$.
		
		\item Up now is the closure of the composition. Let $s, t \in H$. We want to show that $t\cdot s \in H$. We will do this by case analysis.
		
		Case 1: $s$ or $t$ does not exist. W.l.o.g. $s$ is non-existent. Then $s=\stern \in H$ and $t\cdot s = \stern \in H$ because if $t\cdot s$ existed, so would $t$ and $s$.
		
		Case 2: $s$ and $t$ exist but $dom(t) \not \simeq cod(s)$. This implies that there are two different objects in $\cat C$. This means that $t\cdot s = \stern \in \cat C$, the unique non-existing element. Because by definition $s,t \in M$ and $M$ is closed for composition this yields $\stern \in M$. 
		
		But then $\stern \in Hom_M(\stern,\stern)$. Now we show that $\stern \in D(Hom_M(\stern,\stern))$. Note that $End_M(\stern) = \{\stern\}$ and $\stern$ is idempotent. As a result $\stern$ is by default an atom. But because $\stern \leq \stern \, \wedge \, \stern \leq \stern \stern^{-1}$ we get that $\stern \in D(Hom_M(\stern,\stern))$ as desired since the second part of the derivative reduces to what we have just shown. As a result we have $t \cdot s \in H$.
		
		Case 3: $s$ and $t$ exist and $dom(t) \simeq cod(s)$. As a result the composition exists and we can write $s$ as $s: X \to Y$ and $t$ as $t:Y\to Z$ for $X \simeq dom(s), Y \simeq cod(s) \text{ and } Z \simeq cod(t)$. As a result $dom(ts) \simeq dom(s)$ and $cod(ts) \simeq cod(t)$. As a result we want to show that $t\cdot s \in D(Hom_M(X,Z))$. First we will prove
		\begin{align} \begin{split} \label{eq: result}
		\forall \text{ idempotent atoms }& a \in End_M(X) \\
		&\exists h \in Hom_M(X,Z): ((t\cdot s) \leq h \wedge a \leq h^{-1}h). 
		\end{split} \end{align}
		For that fix such an idempotent atom $a \in End_M(X)$. We use the assumptions about $s$ now. That yields \begin{equation} \label{eq: assf} \exists f \in Hom_M(X,Y): s \leq f \, \wedge \, a \leq f^{-1}f. \end{equation} The idea is now to construct something which can be thought of as applying $f$ to $a$ which will be an idempotent atom in $End_M(Y,Y)$. This construction is $fa(fa)^{-1}$.
		
		First note that $fa(fa)^{-1} \simeq faa^{-1}f^{-1} \simeq faf^{-1}$. We get that $dom(faf^{-1}) \simeq dom(f^{-1}) \simeq Y$ and $cod(faf^{-1}) \simeq cod(f) \simeq Y$. As a result $faf^{-1} \in End_M(Y,Y)$.
		
		Next we wish to show that $faf^{-1}$ is an idempotent atom. For idempotence see that $faf^{-1}\cdot faf^{-1} \simeq fa\cdot af^{-1} \simeq faf^{-1}$ by using that from (\ref{eq: assf}) $a \simeq f^{-1}fa$ and the fact that $a$ is idempotent.
		
		In order to show that $faf^{-1}$ is an atom in $End_M(Y,Y)$, we assume $c \leq faf^{-1}$ for $c \in End_M(Y,Y)$ and show that $c \simeq faf^{-1} \vee c \simeq 0$ where $0$ is the zero element of $End_M(Y,Y)$. So let $c \leq faf^{-1}$. Then
		\[ c \leq faf^{-1} \Rightarrow f^{-1}c \leq af^{-1} \Rightarrow f^{-1}cf \leq a. \]
		But because $a$ is an atom by assumption it follows that $f^{-1}cf \simeq a \vee f^{-1}cf \simeq 0$. The later implies that $c \simeq 0$. We prove this by contraposition. 
		
		Suppose $c \not \simeq 0$. $c \simeq faf^{-1}c^{-1}c$ and hence $f^{-1}c \not \simeq 0$ since otherwise $c \simeq 0$. Since $c$ is idempotent by the fact that $faf^{-1}$ is idempotent, $(f^{-1}c)^{-1} \simeq cf$ and by axiom $C1$ of an inverse category $cf \not \simeq 0$. But $cf \simeq faf^{-1}cf$ and as a result $f^{-1}cf \not \simeq 0$.
		
		We may now assume that $f^{-1}cf \simeq a$. As a result $f^{-1}cf$ is an idempotent atom. We wish to show now that $c$ is an inverse of $faf^{-1}$ because this will yield that $c \simeq faf^{-1}$ by the fact that $faf^{-1}$ is its own inverse and as such unique.
		
		It is immediate that \[ faf^{-1} \cdot c \cdot faf^{-1} \simeq faaaf^{-1} \simeq faf^{-1}. \]
		Furthermore,
		\begin{align*}
		& c \cdot faf^{-1} \cdot c \\
		\simeq  \quad & faf^{-1}c \cdot faf^{-1} \cdot faf^{-1}c \\
		\simeq  \quad & faf^{-1}c \cdot faf^{-1}c \\
		\simeq  \quad & faf^{-1}c \\
		\simeq  \quad & c.
		\end{align*}
		
		We conclude that $faf^{-1}$ is indeed an atom, idempotent and an element of $End_M(Y,Y)$. We now use the assumption about $t$ which yields
		\begin{equation} \label{eq: assk} 
		\exists k \in Hom_M(Y,Z): t \leq k \, \wedge \, faf^{-1} \leq k^{-1}k. 
		\end{equation}
		
		We are now in the position to say that we can find $h$ such that (\ref{eq: result}) is satisfied. For this set $h = kf$. Because $s \leq f$ and $t \leq k$ by (\ref{eq: assf}) and (\ref{eq: assk}) respectively we have that $t\cdot s \leq kf$. Furthermore, by (\ref{eq: assk}) it holds that 
		\begin{align}
		& faf^{-1} \simeq k^{-1}kfaf^{-1} \nonumber \\ 
		\Rightarrow \quad & faf^{-1} \simeq faf^{-1}k^{-1}kfaf^{-1} \label{st: 1} \\
		\Rightarrow \quad & af^{-1} \simeq af^{-1}k^{-1}kfaf^{-1} \label{st: 2} \\
		\Rightarrow \quad & a \simeq af^{-1}k^{-1}kfa \label{st: 3} \\
		\Rightarrow \quad & a \leq f^{-1}k^{-1}kf \label{st: 4} \\
		\Rightarrow \quad & a \leq h^{-1}h. \nonumber
		\end{align}
		
		\ref{st: 1} follows by the fact that $faf^{-1}$ is idempotent. \ref{st: 2} and \ref{st: 3} follow because $a = f^{-1}fa$. Then \ref{st: 4} follows because $a$ is idempotent.
		
		Hence we proved (\ref{eq: result}). The second part of the derivative is proven in a similar way and we leave it to the reader to write down the details.
		
		\item What is left to show is that the inclusion property holds in $H$. For this task fix a morphism $s \in \cat C$ with the property that $s \leq t$ for some $t \in H$. Since $M$ is a categorical modeloid and also $t \in M$ we know that $s \in M$. As a result we need to show that 
		\[ s \in D(Hom(dom(t), cod(t))) \] by the fact that $s \leq t$ implies that the domain and codomain of $s$ and $t$ are equal.
		
		Since $t \in H$ we know that  
		\begin{align*} 
		\forall \text{ idempotent atoms }& a \in End_M(dom(t)) \\
		&\exists f \in Hom_M(dom(t),cod(t)): (t \leq f \wedge a \leq f^{-1}f) 
		\end{align*}
		and
		\begin{align*} 
		\forall \text{ idempotent atoms }& b \in End_M(cod(t)) \\
		&\exists k \in Hom_M(dom(t),cod(t)): (t \leq k \wedge b \leq kk^{-1}). 
		\end{align*}
		Because $s \leq t$ and the fact that $\leq$ is a partial order we get that $s \leq f$ and $s \leq k$. But this already implies that $s \in D(Hom(dom(t), cod(t)))$.
	\end{enumerate}
\end{proof}

\subsection{Proofs of section \ref{sc: AlgebraicEFGames}}

\begin{proposition}
	$\cat D := (C,dom,cod,\cdot,\stern,^{-1})$ is an inverse category where $f^{-1}$ denotes the inverse of each partial isomorphism $f$ and $\stern^{-1} = \stern$. The existing elements are exactly all elements in $F(\mathcal{A}, \mathcal{B})$ and the compositions $f \circ g$ in case $dom(f) = cod(g)$ for $f,g \in F(\mathcal{A}, \mathcal{B})$.
\end{proposition}
\begin{proof}
	It is easy to verify that the axioms for a category hold by the constructions of the functions. Using proposition \ref{prop: PropertiesPartBij} it also follows that the axioms additionally required by an inverse category hold.
\end{proof}

\begin{corollary}
	$\cat D := (C,dom,cod,\cdot,\stern,^{-1})$ is also a categorical modeloid on itself.
\end{corollary}
\begin{proof}
	Closure of composition and taking inverses follow by the totality of $\cdot$ and $^{-1}$. The inclusion property and the requirement that all objects of $\cat C$ be in $\cat C$ trivially hold since the modeloid is on itself.
	
	What is to show is that $\cat C$ has a zero element for each endoset $End_C(X)$ where $X$ is an object of $\cat C$. $End_C(X)$ trivially includes the partial isomorphism that is only defined on the constant symbols of $\tau$. We denote it by $0_X$. To see that this is a zero element first note that $\forall p \in End_C(X): 0_X \subseteq p$ by definition of a partial isomorphism. By the definition of partial composition it follows that
	\[ dom(p \circ 0_x) = 0_X^{-1}(cod(0_X) \cap dom (p)) = 0_X^{-1}(dom(0_X)) = dom(0_X). \] As a result we have that $p \circ 0_X = 0_X \, \forall p \in End_C(X)$.
\end{proof}

\begin{theorem}
	Let $M$ be the categorical modeloid $\cat D$. Then
	\[ \exists h:X \to Y \in D^m(M) \text{ with } h \neq \stern \Longleftrightarrow X \equiv_m Y, \quad m \in \mathbb{N} \]
\end{theorem}
\begin{proof}
	'$\Rightarrow$': \quad First we define $ X $ and $ Y $ to be the sets which are uniquely associated to $ dom(h) $ and $ cod(h) $ respectively by the fact that $ dom(h) = id_X $ and $ cod(h) = id_Y $.
	 Then define $I_j := D^j(M) \cap Part(X,Y)$ for $0\leq j\leq m$. We want to prove that $(I_j)_{0\leq j\leq m}$ is a $m$-round winning strategy between $X$ and $Y$. Because $D^m(M) \subseteq D^{m-1}(M) \subseteq ... \subseteq D(M) \subseteq M$ holds all $I_j$ are non-empty and it is clear that $I_j \subseteq Part(X,Y)$ for $0\leq j\leq m$. To show the forth property fix $a\in X$ and some $g\in I_{j+1}$ for $j < m$. Then we know by the definition of $I_{j+1}$ that $dom(g) = X$ and $cod(g) = Y$. Furthermore we have that 
	\[ \forall \text{ idempotent atoms } c\in End_M(X) \exists f\in Hom(X,Y): g \leq f \, \wedge \, c \leq f^{-1}f.  \]
	Similar to proposition \ref{prop: <subsetEq} we have that $s \leq t \Leftrightarrow s \subseteq t$ for $s,t \in Part(X,Y)$. Note that $End_M(X)$ is just $Part(X,X)$ and $Hom_M(X,Y) \subseteq Part(X,Y)$. Now let $0_X$ denote the zero element of $End_M(X)$. We have that $e := 0_X \cup \{(a,a)\}\in End_M(X)$ is an idempotent atom. As a result there is $f\in Hom(X,Y)$ such that $e \subseteq f^{-1}f$. But then we also have that $a \in dom(f)$ and $g \subseteq f$. Hence the forth property holds. The back property follows in a similar way.  
	
	'$\Leftarrow$': \quad Again define $ X $ and $ Y $ to be the sets which are uniquely associated to $ dom(h) $ and $ cod(h) $ respectively. Assume that $(I_j)_{0\leq j\leq m}$ is a $m$-round winning strategy between $X$ and $Y$. Note first that $\forall j \, \forall x: x\in Part(X,Y)$ implies that $x$ exists in terms of free logic. We now want to prove $I_j \subseteq D^j(M)$ by induction on $j$. The base case is clear since $I_0 \subseteq M$. For the induction step take $j \mapsto j+1$. Fix $g\in I_{j+1}$. Then by assumption \[ \forall a\in X \, \exists f\in I_j: g \subseteq f \, \wedge \, a \in dom(f). \] By the induction hypothesis it follows that $f \in D^j(M)$. It is easy to check that the set $E := \{0_X \cup \{(c,c)\} \, | \, c \in X \}$ resembles exactly all idempotent atoms in $End_M(X)$. Hence fix $\hat{a} \in E$. By construction $\hat{a} = 0_X \cup \{(v,v)\}$ for some $v\in X$. Now we know that $\exists f\in I_j: g \subseteq f \, \wedge \, v \in dom(f)$. This yields $g \leq f$ and $\hat{a} \leq f^{-1}f$ again similar to proposition \ref{prop: <subsetEq} and by the fact that $f$ is a partial isomorphism and hence $0_X \subset f^{-1}f$. The second condition of the derivative is shown to be true in a similar way. As a result $g \in D^{j+1}(M)$.
\end{proof}